\documentclass[final,3p]{elsarticle}

\usepackage{amssymb, amsmath}
\usepackage{mathrsfs}
\usepackage{algorithm}
\usepackage{algpseudocode}
\usepackage{hyperref}
\usepackage{braket}
\usepackage{amsthm}

\newtheorem{thm}{Theorem}
\newtheorem{lem}[thm]{Lemma}
\newtheorem{col}[thm]{Corollary}
\theoremstyle{definition}
\newtheorem{mydef}{Definition}
\newtheorem{rmk}{Remark}

\journal{ }

\begin{document}

\begin{frontmatter}



\title{Optimal Separation in Exact Query Complexities for Simon's Problem}


\author{Guangya Cai}
\author{Daowen Qiu\corref{one}}
\cortext[one]{ Corresponding author.\\ \indent{\it E-mail address:} issqdw@mail.sysu.edu.cn (D. Qiu)}

\address{Institute of Computer Science Theory,    School of Data and Computer Science, Sun Yat-sen University, Guangzhou 510006, China}

\begin{abstract}
\emph{Simon's problem} is one of the most important problems demonstrating the power of quantum computers,
which achieves a large separation between quantum and classical query complexities. However, Simon's discussion
on his problem was limited to \emph{bounded-error} setting, which means his algorithm can not always get
the correct answer. \emph{Exact} quantum algorithms for Simon's problem have also been proposed,
which deterministically solve the problem with $O(n)$ queries. Also the quantum lower bound $\Omega(n)$ for Simon's problem is known. Although these algorithms are either complicated
or specialized, their results give an $O(n)$ versus $\Omega(\sqrt{2^{n}})$ separation in exact query
complexities for Simon's problem ($\Omega(\sqrt{2^{n}})$ is the lower bound for classical probabilistic algorithms),
but it has not been proved whether this separation is optimal.
In this paper, we propose another exact quantum algorithm for solving Simon's problem with $O(n)$ queries, which
is simple, concrete and does not rely on special query oracles. Our algorithm combines Simon's algorithm
with the quantum amplitude amplification technique to ensure its determinism.
In particular, we show that Simon's problem can be solved by a classical \emph{deterministic} algorithm with
$O(\sqrt{2^{n}})$ queries (as we are aware, there were no classical deterministic  algorithms for solving Simon's problem with $O(\sqrt{2^{n}})$ queries). Combining some previous results,
we obtain the optimal separation in exact query complexities for Simon's problem:
$\Theta({n})$ versus $\Theta({\sqrt{2^{n}}})$.
\end{abstract}

\begin{keyword}

simon's problem \sep exact query complexity \sep quantum computing
\end{keyword}

\end{frontmatter}


\section{Introduction}
Query complexity has been very useful to study the relative power of quantum computation and classical computation \cite{BW02,MJM11}.
According to their output, query algorithms can be studied either in the \emph{bounded-error} setting (the algorithm
gives the correct result with probability at least $2/3$) or in the \emph{exact} setting (the algorithm gives the correct
result with certainty). For the bounded-error case, there are many algorithms achieving large separation in query complexities
(for example, \cite{Aaronson2015}\cite{Grover1996}), and some of them have exponential speedup for computing partial functions (\cite{Qiu2016} includes a more detailed list),
though it is not known whether the separation is optimal for some of them.

As for exact query complexity, the result is more limited. For total functions, Ambainis \cite{Ambainis2013} gave the first superlinear speedup example,
and the best known separation is $\widetilde{O}(n)$ versus $\Omega(n^{2})$ \cite{Ambainis2016},
which computes a variant of functions introduced in \cite{Goos2015}. In fact, it has been proved that the quantum query algorithms can
only achieve polynomial speedup with degree at most three \cite{Midrijanis2004}. However, for computing partial functions,
there can be an exponential separation, and the first example is the well-known Deutsch-Jozsa problem \cite{Deutsch1992}, whose separation
is 1 versus $n/2+1$. In \cite{Qiu2016}, the optimal separation for a generalized Deutsch-Jozsa problem was given, which is still an exponential one.

Simon's problem \cite{Simon1997} is a famous computational problem that achieves exponential separation in query
complexities. In the bounded-error setting, Simon gave an elegant quantum algorithm which solves the problem
with $O(n)$ queries and the physical realization has demonstrated its efficiency \cite{Tame2014}.
The $\Omega(n)$ lower bound was also proved in \cite{Koiran2007} using polynomial method \cite{Beals2001}.
On the other hand, the classical probabilistic query complexity for this problem is $\Theta(\sqrt{2^{n}})$, which shows that
the $\Theta(n)$ versus $\Theta(\sqrt{2^{n}})$ separation is an optimal one.

 As for the exact query complexities of Simon problem,
Brassard and H{\o}yer \cite{Brassard1997} combined Simon's algorithm with two post-processing subroutines
to ensure that their algorithm solves the problem exactly, which also requires $O(n)$ queries. However, their
algorithm is quite complicated and involved. Mihara and Sung \cite{Mihara2003} proposed a simpler exact
algorithm, but their algorithm relies on some non-standard query oracles and they did not show the construction
of their oracles. Moreover, the $\Omega(n)$ quantum query lower bound is a direct corollary of
previous bounded-error lower bound result. For the classical case, the $\Omega(\sqrt{2^{n}})$ lower bound can be \
easily got (Theorem 6). As we are aware, it is not known whether this lower bound is a tight one,
so it is not known whether the $O(n)$ versus $\Omega(\sqrt{2^{n}})$ is optimal either.

In this paper, we propose a new exact quantum algorithm for solving Simon's problem also with $O(n)$ queries,
which is much simpler and more concrete than Brassard and H{\o}yer's algorithm \cite{Brassard1997} and does not rely on some non-standard
query oracles as Mihara and Sung's construction \cite{Mihara2003}. Our algorithm directly combines Simon's algorithm with the
quantum amplitude amplification technique \cite{Brassard2002} to ensure we get an exact result. Then,
we design a classical deterministic algorithm for solving Simon's problem with $O(\sqrt{2^{n}})$ queries,
which relies on some crucial insights about the bitwise exclusive-or operation results of the pairs of strings
which is queried by the algorithm. Thus, we prove the $\Theta(\sqrt{2^{n}})$ classical deterministic query complexity
for Simon's problem. With previously established results on exact quantum query complexity, we can get the optimal
separation in exact query complexities for Simon's problem: $\Theta(n)$ versus $\Theta(\sqrt{2^{n}})$.

The remainder of the paper is organized as follows. In Section 2, we review Simon's problem, describe the Simon's algorithm \cite{Simon1997} with a different way, and present some notions and notation that will be used hereinafter. Then in Section 3, we discuss the quantum query complexity of Simon's problem and give a new exact quantum algorithm to solve Simon's problem also with $O(n)$ queries. After that, in Section 4, we discuss the classical query complexity of Simon's problem and design a classical deterministic algorithm for solving Simon's problem with $O(\sqrt{2^{n}})$ queries. Finally, conclusions are summarized in Section 5.

\section{Preliminaries}

In the interest of readability, this section serves to introduce some basic notions concerning quantum computation and Simon's problem.

\subsection{Basic introduction to quantum computation}
First, let us introduce some basic terminology of quantum computation. For the details, we can refer to \cite{NC00}.

In quantum computers, the minimal unit of
information is called a \emph{quantum bit} or a \emph{qubit}. As it is known, the classical bit only has a state -
either 0 or 1, but a qubit can be a \emph{superposition} of the two states, written $\Ket{\psi}=\alpha
\Ket{0}+\beta \Ket{1}$. The number $\alpha$ and $\beta$ are complex numbers satisfying $|\alpha|^{2}+|\beta|^{2}=1$,
Put another way, the state of a qubit is a vector in two-dimensional complex vector space. $\Ket{0}$ and
$\Ket{1}$ are known as \emph{computational basis states}, and $\alpha$ and $\beta$ are the \emph{amplitudes} of
the relevant computational basis states.

There are two things we can do with a qubit: measure it or let it evolve unitarily without measurement. We deal
with the measurement first. The most straightforward one is the measurement in the computational basis. In this
way, the measured qubit is either $\Ket{0}$ or $\Ket{1}$. By physical restriction, we do not know the measurement result in advance, but we can ensure that we will see $\Ket{0}$ with probability $|\alpha|^{2}$ and $\Ket{1}$
with probability $|\beta|^{2}$. Of course, there exist other more general kinds of measurement,
but throughout this paper, we only use measurement in the computational basis.

Instead of measuring $\Ket{\psi}$, we can also apply some operations to it. By a complex matrix $U$,
a state $\Ket{\psi}$ can be transformed to a state $\Ket{\varphi}=U\Ket{\psi}$. According to the principle of quantum mechanic \cite{NC00}, the transformation must be a \emph{unitary}
transformation, so $U$ must be a \emph{unitary} matrix.

The notions and notation above describe a system of one qubit, similarly we can think of systems of multiply
qubits. A register of $n$ qubits has $2^{n}$ basis states, each of form
$\Ket{x}=\Ket{x_{1},x_{2},\cdots,x_{n}}=\Ket{x_{1}}\otimes \Ket{x_{2}}\otimes \cdots \otimes\Ket{x_{n}}$,
where $\otimes$ is the tensor product operation and $x\in \{0,1\}^{n}$. The state of the $n$ qubit registers
can be the superposition of the these basis states.
The measurement and state transformation of multiply qubits are similar to the one qubit case as well. Note
that we are also using the tensor product operation to couple the transformation operators on different parts
of the register, i.e. $(A\otimes B)(\Ket{x}\otimes \Ket{y})=A\Ket{x}\otimes B\Ket{y}$.

\subsection{Problem description}
Now, let us recall Simon's problem.
Let $n\geq 1$ be any positive integer and let $(\oplus): \{0,1\}^{n}\times
\{0,1\}^{n}\to \{0,1\}^{n}$ denote the bitwise exclusive-or operation. Suppose we are given a function
$f: \{0,1\}^{n}\to \{0,1\}^{m}$ with $m\geq n$, and we are promised that there exists an
$s\in \{0,1\}^{n}\setminus \{0^{n}\}$ such that for all $x,y\in \{0,1\}^{n}$, $f(x)=f(y)$ if and only if
$x=y$ or $x=y\oplus s$, the aim is to compute $s$.

There exists an associated decision problem for Simon's problem as well.
Suppose that the given function $f$ is either one-to-one, or satisfies the condition defined above. Then the purpose is to determine which of these conditions holds for $f$. Since any lower bound on this problem implies the same one
on the original Simon's problem, it would be useful for the lower bound proof in what follows.

It is now known that Simon's problem is an instance of the \emph{hidden subgroup problem}.
Some notions and notation in group theory  will help our description of quantum algorithms.
Let $G$ denote the group $(\{0,1\}^{n}, \oplus )$. For any subset $X\subseteq G$,
$\langle X \rangle$ denotes the subgroup generated by $X$. The set $X$ is \emph{linearly independent}
if $\langle X \rangle \neq \langle Y \rangle$ for any proper subset $Y$ of $X$. Let $H$ be a subgroup of $G$.
By $H^{\perp}$, we denote the subgroup of $G$ defined by
\begin{equation*}
H^{\perp}=\{g\in G\mid g\cdot h=0 \text{ for all } h\in H\}.
\end{equation*}

We use $(\cdot): \{0,1\}^{n}\times \{0,1\}^{n}\to \{0,1\}$ to denote the inner product modulo 2 of two $n$-bit strings.
Notice that $(H^{\perp})^{\perp}=H$ and $|\langle X \rangle ^{\perp}|=2^{n-|X|}$ if $X$ is linearly independent.

Let us turn our attention back to Simon's problem. At first glance, to solve Simon's problem, one should query a
pair of different strings $x$ and $y$ satisfying $f(x)=f(y)$ and compute $s=x\oplus y$ directly. However, in quantum
computation, Simon took a different step. With the notions and notation described above,
it is easy to show that Simon's problem is equivalent to finding a generation set for subgroup
$K=\{0,s\}$ where $f$ is constant and distinct on each coset of it. Instead of finding $K$ directly,
Simon designed a quantum algorithm finding a generating set of $K^{\perp}$, then computing the generating set of
$K$ with the fact that $(K^{\perp})^{\perp}=K$.

In quantum computation, for a given function $f$, there is an oracle $O_{f}$ that maps
$\Ket{x,y}\to \Ket{x,y\oplus f(x)}$. With this oracle in hand, Simon gave an algorithm that computes $s$
with expected $O(n)$ queries. See Algorithm \ref{alg:simon} for the completed algorithm.
Notice that our description of Simon's algorithm is different from the description in \cite{Simon1997} by
adding some extra examinations of the measured results. These post-processing steps
are to ensure we can get a correct answer in line 13. Also we notice that set $Y$ is not linearly independent
if $|Y|>1$ since $0^{n}\in Y$, but $Y\setminus \{0^{n}\}$ is linearly independent.

Concerning the more introduction to quantum query complexity, we can refer to \cite{BW02}.

\begin{algorithm}
\caption{\strut Simon's algorithm}
\label{alg:simon}
\begin{algorithmic}[1]
\Procedure{Simon}{integer $n$, integer $m$, operator $O_{f}$}\strut
\State $Y\gets \{0^{n}\}$
  \Repeat
  \State Prepare registers $\Ket{0^{n},0^{m}}$
  \State Apply $H^{\otimes n}$ to the first register
  \State Apply $O_{f}$ to the registers
  \State Apply $H^{\otimes n}$ to the first register
  \State Measure the first register, get the result $z$
  \If{$z\notin \langle Y \rangle$} \State $Y\gets Y\cup \{z\}$ \EndIf
  \Until{$|Y|=n$}
\State Find an arbitary $s\in \langle Y \rangle ^{\perp}\setminus \{0^{n}\}$
\State \Return $s$
\EndProcedure \strut
\end{algorithmic}
\end{algorithm}

\section{Exact quantum query complexity for Simon's problem}
Simon's discussion was limited to the bounded-error query. As we can see in the Algorithm \ref{alg:simon}.
If we keep being unlucky, the algorithm might run forever. For example, it is possible to keep
measuring the same result in line 8.
So if only $O(n)$ queries were allowed, we might not get a correct answer.
To design an efficient algorithm in exact query, we add a post-processing subroutine
after line 7 to ensure the loop in Algorithm \ref{alg:simon} runs exactly $n-1$ times, and the $O(n)$ upper
bound is given. On the other hand, one might ask if there exists a more efficient quantum algorithm, the
$\Omega(n)$ lower bound allows the improvement only in constant factor.

\subsection{The lower bound}
To prove the lower bound of quantum query complexity of Simon's problem, Koiran et al.\cite{Koiran2007}
investigated the associated decision problem of Simon's problem described above. Let the algorithm $\mathcal{A}$
be any algorithm that solves Simon's problem exactly since it accepts any bijection function with probability 1
and other functions fulfilling Simon's promise with probability 0. Based on the polynomial method,
Koiran et al. \cite{Koiran2007} first transformed the function describing the probability of $\mathcal{A}$ accepting a given input function into a carefully designed single variable function $Q_{n}(D)$ (defined below)
and prove that the inequality $deg\left (Q_{n}(D) \right ) \leq 2T(n)$ is also sufficient, where
$deg\left (Q_{n}(D) \right)$ denotes the degree of $Q_{n}(D)$ and $T(n)$ denotes the number of queries applied
by $\mathcal{A}$. Next, they proved the following useful lemma.

\begin{lem}[\cite{Koiran2007}]\label{lem:qlb}
Let $Q_{n}(D)$ be the probability that $\mathcal{A}$ accepts $f$ when $f$ is chosen uniformly at random
among the functions $\{0,1\}^{n}\to \{0,1\}^{m}$ hiding a subgroup of order $D$. Then
\begin{equation*}
deg\left (Q_{n}(D) \right )\geq \min \left( \frac{n}{2}, \frac{n+3}{4} \right).
\end{equation*}
\end{lem}

By Lemma \ref{lem:qlb}, it is easy to show the lower bound of quantum query complexity of Simon's problem.

\begin{thm}
Any exact quantum algorithm that solves Simon's problem requires $\Omega(n)$ queries.
\end{thm}

\subsection{The upper bound}
As described in \cite{Brassard1997} and \cite{Mihara2003},
to make Algorithm \ref{alg:simon} exact, we can do this by making sure $(Y\setminus \{0^{n}\})\cup \{z\}$
is always linearly independent when we get the measured result $z$ of the first register.
To accomplish this, Brassard and H{\o}yer \cite{Brassard1997} designed an algorithm with two different steps: (1) They wrote
$K^{\perp}$ as the internal direct sum of two subgroups, one of which is $\langle Y \rangle$, and they gave a
subroutine to guarantee that any eigenstate of the first register is an element of the other group. (2) They gave
a subroutine to preclude $0^{n}$ state of the first register. On the other hand, Mihara and Sung's algorithm \cite{Mihara2003}
relies on their specialized query oracle whose outputs depend on the comparison of different query results of
$f$. Their algorithm is also required to compute an $n$-bit string $w\in \langle Y \rangle ^{\perp}\setminus
\{0^{n}\}$ and they can guarantee that the measured result of the first register is an $n$-bit string
$y\in K^{\perp}$ satisfying $y\cdot w=1$ unless $w=s$.
Compared with their algorithms, our idea is more straight-forward:
making sure the measured result of the first register is in $K^{\perp}$ but not in $\langle Y \rangle$.
The first condition can be fulfilled with Simon's algorithm, and we utilize the quantum amplitude amplification technique
\cite{Brassard2002} to ensure the second one.

Let $\Ket{K^{\perp}, f(T)}$ be the registers states after line 7 of Algorithm \ref{alg:simon}, where
\begin{equation*}
\Ket{K^{\perp}, f(T)}=\frac{1}{2^{n-1}}
\sum_{\substack{x\in T \\ y\in K^{\perp}}}(-1)^{x\cdot y}\Ket{y,f(x)}
\end{equation*}
with $T$ being the set that consists of exactly one representative from each cosets of $K$.
Let $\mathcal{A}=(H^{\otimes n}\otimes I^{\otimes m})O_{f}(H^{\otimes n}\otimes I^{\otimes m})$
denote the combined unitary operators from line 5 to line 7 in Algorithm \ref{alg:simon}.
Moreover, we denote $\mathcal{S}_{0}(\phi)$ and $\mathcal{S}_{\mathcal{A}}(\varphi, Y)$ as follows.
\begin{align*}
\mathcal{S}_{0}(\phi)\Ket{x,b} &=
\begin{cases}
\Ket{x,b}, & x\neq 0^{n} \text{ or } b\neq 0^{m},\\
e^{i\phi}\Ket{x,b}, & x=0^{n} \text{ and } b=0^{m},
\end{cases} \\
\mathcal{S}_{\mathcal{A}}(\varphi, Y)\Ket{x} &=
\begin{cases}
e^{i\varphi}\Ket{x}, & x\notin \langle Y \rangle,\\
\Ket{x}, & x\in \langle Y \rangle.
\end{cases}
\end{align*}
With the definitions of $\mathcal{S}_{0}(\phi)$ and $\mathcal{S}_{\mathcal{A}}(\varphi, Y)$,
we can define quantum amplitude amplification operator as follows
\begin{equation}
\mathcal{Q}=\mathcal{Q}(\mathcal{A},\phi,\varphi, Y)=
-\mathcal{A}\mathcal{S}_{0}(\phi)
\mathcal{A}^{\dagger}(\mathcal{S}_{\mathcal{A}}(\varphi, Y)\otimes I^{\otimes m}).
\end{equation}

In order to make our algorithm work, the crucial step is to eliminate all states in $\langle Y \rangle$ from
the first register. In quantum amplitude amplification process, one can accomplish this by choosing
appropriate $\phi,\varphi\in \mathbb{R}$ such that after applying $\mathcal{Q}$, the amplitude of all states
in $\langle Y \rangle$ of the first register become zero. Let set $X=K^{\perp}\setminus \langle Y \rangle$ and
let $\Ket{K^{\perp},f(T)}=\Ket{\Psi_{X}}+\Ket{\Psi_{Y}}$, where $\Ket{\Psi_{X}}$ denotes
the projection onto the good state subspace (subspace spanned by
$\big\{\Ket{x,b} \mid x\in X, b\in \{0,1\}^{m}\big\}$),
and $\Ket{\Psi_{Y}}$ denotes the projection onto the bad state subspace (subspace spanned
by $\big\{\Ket{y,b} \mid y\in \langle Y \rangle, b\in \{0,1\}^{m} \big\}$).
Notice that $|\langle Y \rangle|=2^{|Y|-1}$ and $|K^{\perp}|=2^{n-1}$. We can obtain the following lemma.

\begin{lem}[\cite{Brassard2002}]
Let $\mathcal{Q}=\mathcal{Q}(\mathcal{A},\phi,\varphi, Y)$. Then
\begin{align*}
\mathcal{Q}\Ket{\Psi_{X}} &= e^{i\varphi}\left((1-e^{i\phi})(1-2^{l-n})-1\right)\Ket{\Psi_{X}} +
                             e^{i\varphi}(1-e^{i\phi})(1-2^{l-n})\Ket{\Psi_{Y}}, \\
\mathcal{Q}\Ket{\Psi_{Y}} &= (1-e^{i\phi})2^{l-n}\Ket{\Psi_{X}} -
                             \left((1-e^{i\phi})(1-2^{l-n})+e^{i\phi} \right)\Ket{\Psi_{Y}}
\end{align*}
where $l=|Y|$.
\end{lem}

By making sure the resulting superposition has inner product zero with $\Ket{\Psi_{Y}}$,
we can obtain the following equation.
\begin{equation}\label{eq:qaa}
e^{i\varphi}(1-e^{i\phi})(1-2^{l-n}) = (1-e^{i\phi})(1-2^{l-n})+e^{i\phi}.
\end{equation}
The chosen $\phi$ and $\varphi$ must satisfy Equation \ref{eq:qaa}. Simple calculation shows that:
\begin{equation}
\phi = 2\arctan{\left(\sqrt{\frac{2^{n-l}}{3\cdot 2^{n-l}-4}}\right)}, \qquad
\varphi = \arccos{\left(\frac{2^{n-l-1}-1}{2^{n-l}-1}\right)}.
\end{equation}
Since $1\leq l\leq n-1$, we can always obtain $\phi,\varphi\in \mathbb{R}$. Thus we get Algorithm \ref{alg:qaa}.

\begin{algorithm}
\caption{\strut Quantum amplitude amplification measuring good states}
\label{alg:qaa}
\begin{algorithmic}[1]
\Require Input parameters satisfy associated definitions \strut
\Ensure $z\in X$
\Procedure{QuAmpAmp}{registers $\Ket{K^{\perp},f(T)}$, integer $n$, operator $\mathcal{A}$, set $Y$}
\State $l\gets |Y|$
\State $\phi\gets 2\arctan{\left(\sqrt{\frac{2^{n-l}}{3\cdot 2^{n-l}-4}}\right)}$
\State $\varphi\gets \arccos{\left(\frac{2^{n-l-1}-1}{2^{n-l}-1}\right)}$
\State Apply $\mathcal{Q}$ to $\Ket{K^{\perp},f(T)}$ where
$\mathcal{Q}=-\mathcal{A}\mathcal{S}_{0}(\phi)
\mathcal{A}^{\dagger}(\mathcal{S}_{\mathcal{A}}(\varphi, Y)\otimes I^{\otimes m})$
\State Measure the first register, get the result $z$
\State \Return $z$
\EndProcedure \strut
\end{algorithmic}
\end{algorithm}

As we can see, the implementation of Algorithm \ref{alg:qaa} requires the constructions of $\mathcal{S}_{0}(\phi)$
and $\mathcal{S}_{\mathcal{A}}(\varphi, Y)$. Here, we discuss the circuit construction and the circuit complexity
of them. It is trivial (in theory) to construct $\mathcal{S}_{0}(\phi)$ and we would use $O(n+m)$ gates for this
task. As for the construction of $\mathcal{S}_{\mathcal{A}}(\varphi, Y)$, the major problem is to determine
whether the given $n$-bit string $x$ is in $\langle Y \rangle$ or not. If $l=1$, we would simply use an $n$ qubits
zero controlled-NOT gate, and if $l>1$, we would solve a system of exclusive-or equations via LU decomposition.
At first glance, it seems that the LU decomposition of the coefficient matrix requires $O(n(l-1)^{2})$ operations.
However, only $O(n(l-1))$ operations are needed since the fact that the LU decomposition result of the previous run
of Algorithm \ref{alg:qaa} can be reused (which also means that all elements of $Y$ should be generated by Algorithm
\ref{alg:qaa} except the initial $0^{n}$). With the LU decomposition result, we need another $O(n(l-1))$ operations
to determine whether the system of equations has a valid solution with $x$ as right-hand side. While the right-hand
side is represented in quantum state, we would use $O(n(l-1))$ controlled-NOT gates for Gauss transformations, $O(l-1)$
swap gates for pivoting, and a final $n-l+1$ qubits zero controlled-NOT gate which the target qubit indicates whether
the solution is a valid one. The total number of gates required to solve the problem is $O(nl)$,
thus we can obtain the $O(nl)$ upper bound for the circuit complexity of $\mathcal{S}_{\mathcal{A}}(\varphi, Y)$.
Then, it is easy to obtain the following lemma.

\begin{lem}
Given $\Ket{K^{\perp},f(T)}$ and a set $Y\subseteq K^{\perp}$ which $\langle Y \rangle \neq K^{\perp}$, there
exits a quantum algorithm that outputs an $n$-bit string $z\in K^{\perp}\setminus \langle Y \rangle$. Moreover,
if the set $Y$ only contains $n$-bit strings generated by that algorithm as well as some constant number initial
elements, the algorithm requires $O(1)$ queries and $O(nl+m)$ gates for other operations, where $l=|Y|$.
\end{lem}

Combining Algorithm \ref{alg:simon} with Algorithm \ref{alg:qaa} in an obvious way, we can get Algorithm
\ref{alg:exactsimon}. Also, it is easy to derive the exact quantum query upper bound of Simon's problem.

\begin{algorithm}
\caption{\strut Exact Simon's algorithm}
\label{alg:exactsimon}
\begin{algorithmic}[1]
\Procedure{ExactSimon}{integer $n$, integer $m$, operator $O_{f}$}\strut
\State $Y\gets \{0^{n}\}$
  \Repeat
  \State Prepare registers $\Ket{0^{n},0^{m}}$
  \State Apply $\mathcal{A}$ to the registers where
         $\mathcal{A}=(H^{\otimes n}\otimes I^{\otimes m})O_{f}(H^{\otimes n}\otimes I^{\otimes m})$
  \State $z\gets$ \Call{QuAmpAmp}{$\Ket{K^{\perp},f(T)}$, $n$, $\mathcal{A}$, $Y$}
  \State $Y\gets Y\cup \{z\}$
  \Until{$|Y|=n$}
\State Find an arbitary $s\in \langle Y \rangle ^{\perp}\setminus \{0^{n}\}$
\State \Return $s$
\EndProcedure \strut
\end{algorithmic}
\end{algorithm}

\begin{thm}
There exists an exact quantum algorithm that solves Simon's problem with $O(n)$ queries.
\end{thm}

\section{Exact classical query complexity for Simon's problem}
In this section, we consider the other side of the exact query complexities separation --- the classical aspect.
We prove the lower bound and upper bound of exact classical query complexity for Simon's problem.
\subsection{The lower bound}
The $\Omega({\sqrt{2^{n}}})$ query complexity lower bound for any randomized algorithm
can be easily adapted from Simon's original paper, and one can see \cite{Kitaev2002} for the details.
This result already implies the query complexity lower bound for any deterministic algorithm.
However, we give a simpler and more concrete proof for the deterministic case here.

\begin{thm}
Any classical deterministic algorithm that solves Simon's problem requires $\Omega({\sqrt{2^{n}}})$ queries.
\end{thm}

\begin{proof}
For classical computers, it is easy to see that any algorithm solving Simon's problem requires
explicitly querying a pair of different $n$-bit strings $x,y\in \{0,1\}^{n}$ satisfying $f(x)=f(y)$.
Because we can not obtain any information of $s$ from the unmatched pair, the theorem can be
reformulated as follow: To find $s$, how many $n$-bit strings should we query at least?

Let $\mathcal{A}$ be an arbitrary deterministic algorithm that solves Simon's problem and
let $Y=\{y^{(1)},y^{(2)},\cdots,y^{(k)}\}$ be the \emph{query set} containing the $n$-bit strings queried
by $\mathcal{A}$. So, given an input function $f$ with a promised $s$, there exits two different $n$-bit
strings $a,b\in Y$ satisfying $f(a)=f(b)$. Equivalently, there exist two different $n$-bit strings
$a,b\in Y$, $s=a\oplus b$.

Consider the \emph{covering set} $S=\{a\oplus b\mid a,b\in Y\}$ generated by $Y$. Next, we will show that
$\{0,1\}^{n}\subseteq S$. It is trivial to see that $0^{n}\in S$ if $Y$ is nonempty. Suppose that there exits an $\tilde{s}\in \{0,1\}^{n}\setminus \{0^{n}\}$ such that $\tilde{s}\notin S$,
we can construct a function $\tilde{f}: \{0,1\}^{n}\to \{0,1\}^{m}$ that $\tilde{f}(x)=\tilde{f}(y)$ if
and only if $x=y$ or $x=y\oplus \tilde{s}$. Of course, the function $\tilde{f}$ is a feasible input to the
algorithm $\mathcal{A}$. So, given $\tilde{f}$ as input, $\mathcal{A}$ must output $\tilde{s}$. However, since
$\tilde{s}\notin S$, $\mathcal{A}$ would be failed to output $\tilde{s}$, that is a contradiction.

We discuss the relationship between $|Y|$ and $|S|$. As $|S|$ is at most $C_{|Y|}^{2}$ and $|S|\geq 2^{n}$
since $\{0,1\}^{n}\subseteq S$, simple calculation shows that $|Y|=\Omega({\sqrt{2^{n}}})$,
which is the lower bound of the cardinality of the query set.
\end{proof}

This lower bound proof also gives us some insights of the upper bound proof. To design an efficient algorithm,
one should find a way to construct a query set whose cardinality is as small as possible but not smaller.
The relationship between query set's and covering set's cardinality should be examined in detail.

\subsection{The upper bound}
First, let us introduce some useful definitions.

\begin{mydef}
Let an $n$-bit string $x\in \{0,1\}^{n}$ with $x=x_{n}x_{n-1}\cdots x_{1}$. The \emph{most significant bit}
of $x$ is $m$ if and only if $x_m=1$ and $x_{i}=0$ for all $i>m$, which is denoted as $MSB(x)=m$. We specify
that $MSB(0^{n})=0$.
\end{mydef}

\begin{rmk}
The definition above is related to the exponent of the scientific notation.
We can view an $n$-bit strings $x$ as an integer in a binary form and write it in a standard form.
Then if the exponent of the integer $x$ is $k$, $MSB(x)=k+1$. In the rest of the paper, we will also
write $n$-bit strings in integer form, for example, $2^{k}$ is equivalent to
$\underbrace{0\cdots 0}_{n-k-1}1\underbrace{0\cdots 0}_{k}$. Also notice that $MSB(2^{k})=k+1$.
\end{rmk}

\begin{mydef}
Let $Y$ be a set containing some $n$-bit strings and let $S=\{a\oplus b\mid a,b\in Y\}$
be its covering set.
If there exists an $m\leq n$ that $0^{n-m}\{0,1\}^{m}\subseteq S$, then $Y$ is called as \emph{$m$-significance}.
Moreover, if $S=0^{n-m}\{0,1\}^{m}$, then $Y$ is \emph{strictly $m$-significance}.
\end{mydef}

As we can see from the lower bound proof, the deterministic algorithm solving Simon's problem
should generate a query set whose covering set contains every $n$-bit string in $\{0,1\}^{n}$,
which means that the query set is $n$-significance.
So, the algorithm may contain the following steps:
\begin{enumerate}[Step 1.]
\item Generate an $n$-significance query set $Y$.
\item Query every $n$-bit string in $Y$ until there are two different $n$-bits strings $a,b\in Y$
satisfying $f(a)=f(b)$.
\item Output $s=a\oplus b$.
\end{enumerate}

The crucial step of the algorithm above is generating a query set whose cardinality is as small as possible.
As we can see, the complexity of the Simon's problem scales with the bits string length $n$. Therefore it is natural to design an algorithm generating the query set for larger $n$ based on the output of lesser $n$.
It is not hard to see that this strategy is equivalent to generating a larger significance query set from
lesser significance query set. The Algorithm \ref{alg:qsg1} utilizes a simple method to accomplish this goal.

\begin{algorithm}
\caption{\strut Preliminary version of algorithm generating query set}
\label{alg:qsg1}
\begin{algorithmic}[1]
\Require $n>1$ \strut
\Ensure $Y$ is $n$-significance
\Procedure{QuerySetGenerator}{integer $n$}
\State $Y\gets \{0^{n}, 0^{n-1}1\}$
  \For{$k\gets 1:n-1$}
    \State $Z\gets \emptyset$
      \ForAll{$y\in Y$ and $MSB(y)<k$}
        \State $Z\gets Z\cup \{2^{k}\oplus y\}$
      \EndFor
    \State $Y\gets Y\cup Z$
  \EndFor
\State \Return $Y$
\EndProcedure \strut
\end{algorithmic}
\end{algorithm}

Before we prove the correctness of the algorithm, let us define some notations.
Let $Y^{(0)}$ denote the set $Y$ in line 2, and let $Y^{(1)}, Y^{(2)}, \cdots, Y^{(n-1)}$ denote
set $Y$ in the left-hand side of line 8 for each loop.
Similarly, $Z^{(1)}, Z^{(2)}, \cdots, Z^{(n-1)}$ also denote the set $Z$ of line 8 for each loop.
It is not hard to show the following invariants for Algorithm \ref{alg:qsg1}.
\begin{equation}\label{eq:loopinv1}
\max_{y\in Y^{(k)}}(MSB(y))=k+1\text{ for all } k\geq 0,
\end{equation}
\begin{equation}\label{eq:loopinv7}
MSB(z)=k+1\text{ for all } z\in Z^{(k)} \text{ and } k\geq 1,
\end{equation}
\begin{equation}\label{eq:loopinv2}
Z^{(k)}=\{2^{k}\oplus y\mid y\in Y^{(k-2)}\}\text{ for all } k\geq 2.
\end{equation}
Equation \ref{eq:loopinv2} can be easily derived from Equation \ref{eq:loopinv1} and Equation \ref{eq:loopinv7}.

For the correctness of the algorithm, $Y^{(n-1)}$ must be $n$-significance.
However, in Algorithm \ref{alg:qsg1}, each $Y^{(k)}$ satisfies a stronger condition.

\begin{thm}\label{thm:alg1}
In Algorithm \ref{alg:qsg1}, $Y^{(k)}$ is $k$+1-significance for all $k\geq 0$.
\end{thm}

\begin{proof}
We prove by induction on $k$. It is easy to see that $Y^{(0)}$ is 1-significance and $Y^{(1)}$ is 2-significance
(Notice that $Y^{(0)}=\{0^{n},0^{n-1}1\}$ and $Y^{(1)}=Y^{(0)}\cup \{0^{n-2}10\}$).
Assume there is  a $j>1$ satisfying that $Y^{(l)}$ is $l$+1-significance for all $l\leq j$. Next, we show that
$Y^{(j+1)}$ is $j$+2-significance.

Let $S$ be the covering set of $Y^{(j+1)}$. In order to show $Y^{(j+1)}$ is $j$+2-significance,
$0^{n-j-2}\{0,1\}^{j+2}\subseteq S$ must be satisfied. By line 8, $Y^{(j)}\subseteq Y^{(j+1)}$,
which means $0^{n-j-2}0\{0,1\}^{j+1}\subseteq S$. To prove the theorem,
we only need to show that $0^{n-j-2}1\{0,1\}^{j+1}\subseteq S$.

Also by line 8, we can see that $Y^{(j+1)}=Y^{(j)}\cup Z^{(j+1)}=Y^{(j-1)}\cup Z^{(j)}\cup Z^{(j+1)}$.
Let set $S'=\{a\oplus b\mid a\in Z^{(j+1)}, b\in Y^{(j-1)}\}$ and set
$S''=\{a\oplus b\mid a\in Z^{(j+1)}, b\in Z^{(j)}\}$. Notice that $S'\subseteq S$ and $S''\subseteq S$.

Consider the set $S'$ first. By Equation \ref{eq:loopinv2},
$S'=\{2^{j+1}\oplus (y\oplus b)\mid y,b\in Y^{(j-1)}\}$, and it is clear to see that
$0^{n-j-2}10\{0,1\}^{j}\subseteq S'$.

As for the set $S''$, also by Equation \ref{eq:loopinv2},
$S''=\{2^{j+1}\oplus (y\oplus b)\mid y\in Y^{(j-1)}, b\in Z^{(j)}\}$.
Consider the set $X=\{y\oplus b\mid y\in Y^{(j-1)}, b\in Z^{(j)}\}$.
Notice that $Y^{(j)}=Y^{(j-1)}\cup Z^{(j)}$. With Equation \ref{eq:loopinv1} and Equation \ref{eq:loopinv7},
the forthcoming Lemma \ref{lem:lem1} shows that $0^{n-j-1}1\{0,1\}^{j}\subseteq X$. Therefore,
it is clear to see that $0^{n-j-2}11\{0,1\}^{j}\subseteq S''$.

Because $S'$ and $S''$ are subsets of $S$, $0^{n-j-2}1\{0,1\}^{j+1}\subseteq S$ is proved, which leads to the
result that $Y^{(j+1)}$ is $j$+2-significance.
\end{proof}

The proof of Theorem \ref{thm:alg1} relies on the following lemma.

\begin{lem}\label{lem:lem1}
Let the set $Y$ be a set containing some $n$-bit strings. If $Y$ fulfils the following conditions:
\begin{enumerate}[(1)]
\item $MSB(y)\leq k$ for all $y\in Y$,
\item $Y$ is $k$-significance,
\end{enumerate}
then by denoting $A=\{a\mid a\in Y, MSB(a)=k\}$, $B=\{b\mid b\in Y, MSB(b)<k\}$, we have that
set $X=\{a\oplus b\mid a\in A, b\in B\}$ satisfies $X\supseteq 0^{n-k}1\{0,1\}^{k-1}$.
\end{lem}

\begin{proof}
Notice that the set $X$ is a subset of the covering set $S$ of $Y$. To complete the proof, we must show that for every $s\in S$
with $s\in 0^{n-k}1\{0,1\}^{k-1}$, $s$ must be in $X$. Suppose a contradiction that there exists an $\tilde{s}\in 0^{n-k}1\{0,1\}^{k-1}$
that $\tilde{s}\notin X$. Then $\tilde{s}$ must be in $S\setminus X$. Notice that $Y=A\cup B$. Then $S\setminus X=C\cup D$,
with $C=\{x\oplus y\mid x, y\in A\}$ and $D=\{x\oplus y\mid x, y\in B\}$. Since $MSB(s)<k$ for all $s\in C$ or $s\in D$ and
$MSB(\tilde{s})=k$, that is a contradiction.
\end{proof}

The correctness of the Algorithm \ref{alg:qsg1} is a direct result of Theorem \ref{thm:alg1}.
\begin{thm}\label{thm:correct1}
Algorithm \ref{alg:qsg1} correctly constructs an $n$-significance query set with input $n$.
\end{thm}

Although Theorem \ref{thm:alg1} is sufficient enough to prove the correctness of the Algorithm \ref{alg:qsg1},
we have to state that each $Y^{(k)}$ in Algorithm \ref{alg:qsg1} actually satisfies an even stronger condition.

\begin{col}
In Algorithm \ref{alg:qsg1}, $Y^{(k)}$ is strictly $k$+1-significance for all $k\geq 0$.
\end{col}

\begin{proof}
Let $S$ be the covering set of $Y^{(k)}$. By Equation \ref{eq:loopinv1}, it is easy to see that $MSB(s)\leq k+1$ for all $s\in S$,
which indicates that $S \subseteq 0^{n-k-1}\{0,1\}^{k+1}$. By Theorem \ref{thm:alg1}, $0^{n-k-1}\{0,1\}^{k+1} \subseteq S$.
Consequently $S=0^{n-k-1}\{0,1\}^{k+1}$, and this completes the proof.
\end{proof}

We have proved Theorem \ref{thm:correct1}, but there remains another question: is it an efficient one from the viewpoint of complexity? By Equation \ref{eq:loopinv2},
it is easy to see that $|Z^{(k)}|=|Y^{(k-2)}|$ for all $k\geq 2$, and therefore we can obtain a simple relationship that
$|Y^{(k)}|=|Y^{(k-1)}|+|Y^{(k-2)}|$. As a result, the cardinality of the query set scales as a Fibonacci sequence with increasing $n$.
Though this algorithm is more efficient than the naive one, it is not efficient enough to match the proved lower bound.
Notice that $|Z^{(k)}|$ also scales as a Fibonacci sequence. If we design a more efficient algorithm,
the number of new $n$-bit strings added to set $Y$ should be lesser. How about $|Z^{(k)}|\approx \sqrt{2}|Z^{(k-1)}|$?
Such an algorithm will be efficient enough, but the problem is to make it a correct one.
Algorithm \ref{alg:qsg2} accomplishes a similar goal.

\begin{algorithm}
\caption{\strut Final version of algorithm generating query set}
\label{alg:qsg2}
\begin{algorithmic}[1]
\Require $n>2$ \strut
\Ensure $Y$ is $n$-significance
\Procedure{QuerySetGenerator}{integer $n$}
\State $Y\gets \{0^{n}, 0^{n-1}1, 0^{n-2}10\}$
  \For{$k\gets 1:n-2$}
    \State $Z\gets \emptyset$
      \ForAll{$y\in Y$ and $MSB(y)=k$}
        \State $Z\gets Z\cup \{2^{k+1}\oplus y, 2^{k+1}\oplus 2^{k-1}\oplus y\}$
      \EndFor
    \State $Y\gets Y\cup Z$
  \EndFor
\State \Return $Y$
\EndProcedure \strut
\end{algorithmic}
\end{algorithm}

Just as Algorithm \ref{alg:qsg1}, the definitions of notation $Y^{(k)}$ and $Z^{(k)}$ are similar, only without $Y^{(n-1)}$ and $Z^{(n-1)}$
(the outer loop in Algorithm \ref{alg:qsg2} runs $n-2$ times). Also, it is not hard to see the following invariants for Algorithm \ref{alg:qsg2}.
\begin{equation}\label{eq:loopinv3}
\max_{y\in Y^{(k)}}(MSB(y))=k+2\text{ for all } k\geq 0,
\end{equation}
\begin{equation}\label{eq:loopinv6}
MSB(z)=k+2\text{ for all } z\in Z^{(k)} \text{ and } k\geq 1,
\end{equation}
\begin{equation}\label{eq:loopinv4}
Z^{(k)}=\{2^{k+1}\oplus z, 2^{k+1}\oplus 2^{k-1}\oplus z\mid z\in Z^{(k-2)}\} \text{ for all } k\geq 3,
\end{equation}
\begin{equation}\label{eq:loopinv5}
y_{k-1}=0\text{ for all }y\in Y \text{; }MSB(y)=k \text{ and } k\geq 2.
\end{equation}

Before discussing the correctness of Algorithm \ref{alg:qsg2}, we take a brief detour
to analyze its efficiency, which relates to the cardinality of the query set
constructed by Algorithm \ref{alg:qsg2}.

\begin{thm}\label{thm:qsgeff}
Algorithm \ref{alg:qsg2} constructs a query set whose cardinality is $O(\sqrt{2^{n}})$.
\end{thm}

\begin{proof}
For a given input $n$, it is easy to see that the cardinality of the constructed query set is
$|Y|=|Y^{(n-2)}|=|Z^{(n-2)}|+|Z^{(n-1)}|+\cdots +|Z^{(1)}|+|Y^{(0)}|$. Moreover, by Equation \ref{eq:loopinv4},
$|Z^{(k+2)}|=2|Z^{(k)}|$. With the initial values $|Z^{(1)}|=2$, $|Z^{(2)}|=2$ and $|Y^{(0)}|=3$,
for a large enough $n$, simple calculation shows that:
\begin{equation*}
\begin{split}
|Y|&=\sum_{k=1}^{n-2}|Z^{(k)}|+|Y^{(0)}| \\
   &=\sum_{\text{odd } k}2^{\frac{k+1}{2}}+\sum_{\text{even }k}2^{\frac{k}{2}}+3 \\
   &=2(2^{\lceil \frac{n-2}{2} \rceil }-1)+
     2(2^{\lfloor \frac{n-2}{2} \rfloor }-1)+3 \\
   &=2^{\lceil \frac{n}{2} \rceil}+2^{\lfloor \frac{n}{2} \rfloor}-1 \\
   &=O(\sqrt{2^{n}}).
\end{split}
\end{equation*}
\end{proof}

\begin{rmk}
Though it is not well as the case of odd inputs, $|Y|\approx 2\sqrt{2^{n}}$ is still a good fit even with a large $n$.
This also implies that our algorithm is not only efficient in the asymptotic sense, but also not a bad one
considering the constant factor.
\end{rmk}

Next, we will show that Algorithm \ref{alg:qsg2} is a correct one. Similar to Algorithm \ref{alg:qsg1},
the correctness of the algorithm depends on some properties of $Y^{(k)}$.

\begin{thm}\label{thm:qsg2}
In Algorithm \ref{alg:qsg2}, $Y^{(k)}$ is $k$+2-significance for all $k\geq 0$.
\end{thm}

\begin{proof}
We proceed with induction on $k$. It is simple to examine the cases that $k=0,1,2$
(Notice that $Y^{(0)}=\{0^{n},0^{n-1}1,0^{n-2}10\}$, $Y^{(1)}=Y^{(0)}\cup \{0^{n-3}101,0^{n-3}100\}$ and
$Y^{(2)}=Y^{(1)}\cup \{0^{n-4}1010,0^{n-4}1000\}$).
Suppose for a $j>2$ that $Y^{(l)}$ is $l$+2-significance for all $l\leq j$.
Next we will show that $Y^{(j+1)}$ is $j$+3-significance.

Let $S$ be the covering set of $Y^{(j+1)}$. Since $Y^{(j+1)}$ is $j$+2-significance by line 8,
we only need to show that
$0^{n-j-3}1\{0,1\}^{j+2}\subseteq S$ so  as to show that $Y^{(j+1)}$ is also $j$+3-significance.

Also by line 8, we can see that $Y^{(j+1)}=Y^{(j)}\cup Z^{(j+1)}=Y^{(j-2)}\cup Z^{(j-1)}\cup
Z^{(j)}\cup Z^{(j+1)}$. Let set $A=\{2^{j+2}\oplus a\mid a\in Y^{(j)}, MSB(a)=j+1\}$,
set $B=\{2^{j+2}\oplus 2^{j}\oplus b\mid b\in Y^{(j)}, MSB(b)=j+1\}$, and let sets $S^{(0)}$, $S^{(1)}$,
$S^{(2)}$, $S^{(3)}$ be defined as below. Notice that $Z^{(j+1)}=A\cup B$ and $S^{(0)}$, $S^{(1)}$,
$S^{(2)}$, $S^{(3)}$ are subsets of $S$.
\begin{align*}
S^{(0)}&=\{a\oplus y\mid a\in A, y\in Y^{(j-2)}\}, &
S^{(1)}&=\{b\oplus y\mid b\in B, y\in Y^{(j-2)}\}, \\
S^{(2)}&=\{a\oplus z\mid a\in A, z\in Z^{(j)}\}, &
S^{(3)}&=\{b\oplus z\mid b\in B, z\in Z^{(j)}\}.
\end{align*}

Consider the sets $S^{(0)}$ and $S^{(1)}$ first. By Equation \ref{eq:loopinv4}, it is easy to see that
$S^{(0)}=\{2^{j+2}\oplus (z\oplus y)\mid z\in Z^{(j-1)}, y\in Y^{(j-2)}\}$ and
$S^{(1)}=\{2^{j+2}\oplus 2^{j}\oplus (z\oplus y)\mid z\in Z^{(j-1)}, y\in Y^{(j-2)}\}$.
Let set $X=\{z\oplus y\mid z\in Z^{(j-1)}, y\in Y^{(j-2)}\}$. With Equation \ref{eq:loopinv3}, Equation \ref{eq:loopinv6}
and the fact that $Y^{(j-1)}=Y^{(j-2)}\cup Z^{(j-1)}$, Lemma \ref{lem:lem1} shows that
$0^{n-j-1}1\{0,1\}^{j}\subseteq X$. It is clear to see $0^{n-j-3}101\{0,1\}^{j}\subseteq S^{(0)}$ and
$0^{n-j-3}100\{0,1\}^{j}\subseteq S^{(1)}$.

Next, consider the sets $S^{(2)}$ and set $S^{(3)}$. Also by Equation \ref{eq:loopinv4},
$S^{(2)}=\{2^{j+2}\oplus (w\oplus z)\mid w\in Z^{(j-1)}, z\in Z^{(j)}\}$ and
$S^{(3)}=\{2^{j+2}\oplus 2^{j}\oplus (w\oplus z)\mid w\in Z^{(j-1)}, z\in Z^{(j)}\}$.
Let set $U=\{w\oplus z\mid w\in Z^{(j-1)}, z\in Z^{(j)}\}$. With Equation \ref{eq:loopinv6},
Equation \ref{eq:loopinv5} and the fact that $Z^{(j-1)}\subseteq Y^{(j)}$ and $Z^{(j)}\subseteq Y^{(j)}$,
the forthcoming Lemma \ref{lem:lem2} shows that $0^{n-j-2}11\{0,1\}^{j}\subseteq U$.
It is clear to see $0^{n-j-3}111\{0,1\}^{j}\subseteq S^{(2)}$ and $0^{n-j-3}110\{0,1\}^{j}\subseteq S^{(3)}$.

Collecting the previous results, we can see that $S^{(0)}\cup S^{(1)}\cup S^{(2)}\cup S^{(3)}\supseteq
0^{n-j-3}1\{0,1\}^{j+2}$. Thus $Y^{(j+1)}$ is $j$+3-significance.
\end{proof}

To complete the proof of Theorem \ref{thm:qsg2}, we have to prove the following lemma.

\begin{lem}\label{lem:lem2}
Let set $Y$ be a set containing some $n$-bit strings. If $Y$ fulfils the following conditions:
\begin{enumerate}[(1)]
\item $MSB(y)\leq k$ for all $y\in Y$,
\item $Y$ is $k$-significance,
\item $y_{k-1}=0$ for all $y\in Y$ and $MSB(y)=k$,
\end{enumerate}
then by denoting set $D=\{d\mid d\in Y, MSB(d)=k\}$, $E=\{e\mid e\in Y, MSB(e)=k-1\}$, we have that the
set $U=\{d\oplus e\mid d\in D, e\in E\}$ satisfies $U\supseteq 0^{n-k}11\{0,1\}^{k-2}$.
\end{lem}

\begin{proof}
Notice that set $U$ is also a subset of the covering set $S$ of $Y$. Suppose a contradiction that
there exists an $\tilde{s}\in 0^{n-k}11\{0,1\}^{k-2}$ that $\tilde{s}\notin U$.
Then $\tilde{s}$ must be in $S\setminus U$. Let set $W=\{w\mid w\in Y, MSB(w)<k-1\}$. It is easy to
see that $Y=D\cup E\cup W$. Because $MSB(\tilde{s})=k$ and $\tilde{s}\notin U$, $\tilde{s}$
must be in set $V$ with $V=\{d\oplus w\mid d\in D, w\in W\}$. However, for every element $v\in V$,
$v_{k-1}=0$. Since $\tilde{s}_{k-1}=1$, that is a contradiction.
\end{proof}

The correctness of Algorithm \ref{alg:qsg2} is also a direct result of Theorem \ref{thm:qsg2}. Moreover,
similar to Algorithm \ref{alg:qsg1},
$Y^{(k)}$ in Algorithm \ref{alg:qsg2} actually satisfies a stronger condition.
\begin{thm}\label{thm:correct2}
Algorithm \ref{alg:qsg2} correctly constructs an $n$-significance query set with input $n$.
\end{thm}
\begin{col}\label{col:strict}
In Algorithm \ref{alg:qsg2}, $Y^{(k)}$ is strictly $k$+2-significance for all $k\geq 0$.
\end{col}

Combining Theorem \ref{thm:qsgeff} and Theorem \ref{thm:correct2}, we can easily obtain
the upper bound of exact classical query complexity for Simon's problem.
\begin{thm}
There exists a classical deterministic algorithm that solves Simon's problem with $O(\sqrt{2^{n}})$ queries.
\end{thm}

\begin{rmk}
Recall that running Algorithm \ref{alg:qsg2} just finishes  the first step of our whole algorithm, and
the second step does the queries. Since the cardinality of the query set is $O(\sqrt{2^{n}})$,
it is simple to see that the upper bound of query complexity is also $O(\sqrt{2^{n}})$.
\end{rmk}

\section{Concluding remarks}

It is believed that quantum computers are likely more powerful than classical computers, and Simon's problem is
one of the instances that support this point of view. By proving an optimal separation in exact query
complexities for Simon's problem, we have tested the power of exact quantum computation for a certain
problem and a certain computational complexity.

However, our separation is optimal only up to a constant factor, and one may ask for an even tighter one.
Both our exact quantum query algorithm and Brassard and H{\o}yer's one \cite{Brassard1997} require $3n-3$ queries, and the constant factor is bigger
than that in the proved lower bound by \cite{Koiran2007}. Mihara and Sung's algorithm \cite{Mihara2003} only requires $n-1$ invocations of their oracle, but
 they did not show its construction, and thus we can not count the actual invocations of $O_{f}$.

Also, we have not known whether the covering set constructed by Algorithm \ref{alg:qsg2} is the smallest one.
In fact, Corollary \ref{col:strict} imposes a strong condition on that algorithm, so one might utilize
a weaker condition to design a more efficient algorithm. Such problems may be worthy of further investigations.

\section*{Acknowledgements}
This work is supported in
part by the National Natural Science Foundation of China (Nos.  61572532, 61272058).








\end{document}